\documentclass[11pt]{llncs}
\usepackage{fullpage}
\usepackage{amsmath}
\usepackage{latexsym}
\usepackage{xspace}
\usepackage{subfig}
\usepackage{graphicx}
\usepackage{color}
\usepackage{alltt}
\usepackage{url}
\usepackage{color}
\usepackage{epsfig}

\newtheorem{thrm}{Theorem}
\newtheorem{lemm}[thrm]{Lemma}

\def\itbf#1{\textit{\textbf{#1}}}
\def\str#1{\texttt{#1}}
\def\Oh{\mathcal{O}}

\def\s#1{\mbox{\boldmath $#1$}}

\def\pref#1{\mbox{pref(\s{#1})}}
\def\suff#1{\mbox{suff(\s{#1})}}

\def\:{\mbox{\ :\ }}

\def\+{\!+\!}
\def\-{\!-\!}

\def\rank{\mbox{\rm rank}}

\def\balgorithm#1{{\bf Algorithm #1}}

\def\bfor{{\bf for\ }}

\def\bdownto{{\bf downto\ }}
\def\bwhile{{\bf while\ }}

\def\bbreak{{\bf break\ }}

\def\bdo{{\bf do\ }}
\def\bif{{\bf if\ }}
\def\bthen{{\bf then\ }}
\def\belse{{\bf else\ }}

\def\breturn{{\bf return\ }}

\def\la{\leftarrow}

\def\rem#1{\hspace{24pt}{\sl #1}}
\def\pref(#1,#2){$#1$ is a prefix of $#2$}
\def\suff(#1,#2){$#1$ is a suffix of $#2$}
\def\reg(#1,#2){$#2$ is $#1$-regular}
\def\notreg(#1,#2){$#2$ is not $#1$-regular}

\def\Mbwt{\mathcal{M}}

\newcommand{\dna}{{\sc dna}\xspace}
\newcommand{\english}{{\sc english}\xspace}
\newcommand{\xml}{{\sc xml}\xspace}
\newcommand{\sources}{{\sc source}\xspace}

\begin{document}

\title{Fixed Block Compression Boosting in FM-Indexes\thanks{This work
    was supported by the Academy of Finland grant 118653 (ALGODAN) and
    by the Australian Research Council. Simon J. Puglisi is
    supported by a Newton Fellowship.}}

\author{ Juha K{\"a}rkk{\"a}inen\inst{1} \and Simon J.\
  Puglisi\inst{2} }

\institute{
  Department of Computer Science,
  University of Helsinki, Finland\\
  \email{juha.karkkainen@cs.helsinki.fi}\\[1ex]
  \and
  Department of Informatics,
  King's College London, London, United Kingdom\\
  \email{simon.puglisi@kcl.ac.uk} }

\maketitle

\thispagestyle{empty}

\begin{abstract}
  A compressed full-text self-index occupies space close to that of
  the compressed text and simultaneously allows fast pattern matching
  and random access to the underlying text. Among the best compressed
  self-indexes, in theory and in practice, are several members of the
  FM-index family.  In this paper, we describe new FM-index variants
  that combine nice theoretical properties, simple implementation and
  improved practical performance. Our main result is a new technique
  called {\em fixed block compression boosting}, which is a simpler
  and faster alternative to optimal compression boosting and implicit
  compression boosting used in previous FM-indexes.
\end{abstract}

\section{Introduction}
\label{sect-intro}

A compressed full-text self-index~\cite{nm2007} of a text string $T$
is a data structure that stores $T$ in a compressed form that allows
fast random access to $T$ and also supports fast pattern matching
queries. We focus here on the \emph{count} query that, given a pattern
string $P$, returns the number of occurrences of $P$ in
$T$.~\footnote{Our indexes support other common queries such as
  \emph{locate} and \emph{extract}, but the algorithmic and
  implementation issues in engineering them are quite different and
  outside the scope of this paper.}  Many of the best compressed
self-indexes, in theory and in practice, belong to the FM-index family
originating from the FM-index of Ferragina and
Manzini~\cite{fm2005}. In particular, they combine good compression
with fast count
queries~\cite{fmmn2007,mn2007-implicit,fgnv2009,cn2008}. In this
paper, we describe new variants of the FM-family achieving even better
compression and faster count queries.

The main components of most FM-indexes are:
\begin{itemize}
\item The \emph{Burrows--Wheeler transform} (BWT)~\cite{bw1994}: an
  invertible permutation of the text~$T$.  A procedure called
  \emph{backward search}~\cite{fm2005} turns a count query on $T$ into
  a sequence of \emph{rank} queries on the BWT.
\item The \emph{wavelet tree}~\cite{ggv2003}: a representation of the
  BWT that turns a BWT rank query into a sequence of rank queries on
  bitvectors.
\item A bitvector \emph{rank index}, which supports fast rank queries
  on bitvectors.
\end{itemize}

The total length of standard wavelet tree bitvectors is equal to the
size of the original, uncompressed text in bits.  All other data
structures can be fitted in less space: asymptotically less in theory,
and significantly less in practice.  Basic zero-order compression is
achieved either with compressed bitvector rank structures, such as
RRR~\cite{rrr2007}, or Huffman-shaped wavelet trees~\cite{ggv2004}.
For higher order compression, we can use a technique called
\emph{compression boosting}~\cite{fgms2005,fmmn2007}, where the BWT is
partitioned into blocks of varying sizes based on the context of
symbols in $T$, and there is a separate, zero-order compressed wavelet
tree for each block.  An optimal partitioning into context blocks can
be found in linear time~\cite{fgms2005}.

Our main result is a technique called {\em fixed block compression
  boosting}. It is similar to context block boosting, but divides the
BWT into blocks of fixed sizes without any regard to the symbol
contexts. Such a division is inoptimal, but we show that it cannot be
much worse than the optimal one. What we gain is simpler and faster
data structures. The difference is particularly dramatic in the
construction phase.

The RRR-structure for compressed bitvector rank~\cite{rrr2007} divides
the bitvectors into small blocks of fixed sizes. M\"akinen and
Navarro~\cite{mn2007-implicit} show that this achieves a similar
compression boosting effect without any explicit division of the BWT.
This is called \emph{implicit compression boosting}.  Their analysis
of the effect of fixed blocks inspired our analysis, but the extension
from small blocks on bitvectors to larger blocks and larger alphabets
is non-trivial.

There are implementations of FM-indexes without any compression
boosting~\cite{fgnv2009}, with optimal context block
boosting~\cite{fgnv2009}, and with implicit boosting~\cite{cn2008}.
Fixed block boosting has practical advantages over all these, which we
demontrate experimentally.

\section{Basic Algorithmic Machinery}

Let $T = T[0..n-1] = T[0]T[1]\ldots T[n-1]$ be a string of $n$ symbols
or characters drawn from an alphabet $\Sigma=\{0,1,..,\sigma-1\}$.  We
assume that $T[n-1]=0$ and $0$ does not appear anywhere else in
$T$. In the examples, we use `\$' to denote 0 and letters to denote
other symbols.

For any $i\in 0..n-1$, the string $T[i..n-1]T[0..i-1]$ is a
\emph{rotation} of $T$. Let $\Mbwt$ be the $n \times n$ matrix whose
rows are all the rotations of $T$ in lexicographic order. Let $F$ be
the first and $L$ the last column of $\Mbwt$, both taken to be strings
of length $n$. The string $L$ is the \itbf{Burrows--Wheeler transform}
of $T$. An example is given in Fig.~\ref{fig:BWT}. Note that $F$ and
$L$ are permutations of $T$.

\begin{figure}
  \centering

{\tt
  \begin{tabular}{|c|ccccc|c|}
    \multicolumn{1}{c}{$F$} &&&&&\multicolumn{1}{c}{}& \multicolumn{1}{c}{$L$}\\
    \cline{1-1}\cline{7-7}
    $\$$&B&A&N&A&N&A\\
    A&$\$$&B&A&N&A&N\\
    A&N&A&$\$$&B&A&N\\
    A&N&A&N&A&$\$$&B\\
    B&A&N&A&N&A&$\$$\\
    N&A&$\$$&B&A&N&A\\
    N&A&N&A&$\$$&B&A\\
    \cline{1-1}\cline{7-7}
  \end{tabular}
}
\caption{BWT matrix $\Mbwt$ for text $T=\str{BANANA\$}$.}
\label{fig:BWT}
\end{figure}

The FM-family of compressed text self-indexes is based on a procedure
called \itbf{backward search}, which finds the range of rows in
$\Mbwt$ that begin with a given pattern $P$. This range represents the
occurrences of $P$ in $T$. Fig.~\ref{fig:count} shows how backward
search is used for implementing the count query.  In the algorithm,
$C[c]$ is the position of the first occurrence of the symbol $c$ in
$F$, and the function $\rank_L$ is defined as
\[
\rank_L(c,j) \equiv \big|\{i\mid i<j \mbox{ and } L[i]=c\}\big|
\]
The main difference between the members of the FM-family is how they
implement the $\rank_L$-function. The best ones use wavelet trees.

\begin{figure}
  \centering
  \begin{tabbing}
    10:\=\qquad\=\qquad\=\qquad\=\kill
    \rule{\linewidth}{1pt}\\
    \balgorithm{FM-Count($P[0..m-1]$)}\\
    1:\>$b \la 0$; $e \la n$ \\
    2:\>\bfor $i \la m-1$ \bdownto $0$ \bdo\\
    3:\>\>$c \la P[i]$\\
    4:\>\>$b \la C[c] + \rank_L(c,b)$\\
    5:\>\>$e \la C[c] + \rank_L(c,e)$\\
    6:\>\>\bif $b=e$ \bthen \bbreak \rem{//The range is empty}\\
    7:\>\breturn $e-b$ \rem{//The range is $b..e-1$}\\
    \rule{\linewidth}{1pt}
  \end{tabbing}
  \caption{Counting pattern occurrences using backward search.}
  \label{fig:count}
\end{figure}

A wavelet tree of a string $X$ over an alphabet $\Sigma$ is a binary
tree with leaves labelled by the symbols of $\Sigma$.  Each node $v$
is associated with the subsequence of $X$ consisting of those symbols
that appear in the subtree rooted at $v$. The associated strings are
not stored; instead each internal node $v$ stores a bitvector $B(v)$
that tells for each character in the associated string whether it is
in the left or right subtree of $v$. Fig.~\ref{fig:wt} shows examples
of the two commonly used variants of wavelet trees, the balanced and
the Huffman-shaped.

The balanced wavelet tree is easy to
implement with low overhead. The total length of the bitvectors is
$|X|\lceil\log|\Sigma|\rceil$, which is exactly the length of $X$ in
bits using the standard representation. On the other hand, the
Huffman-shaped wavelet tree (HWT) is the one that minimizes the total
length of the bitvectors, which equals the size of the Huffman
compressed string $X$.

\begin{figure}
  \centering \input{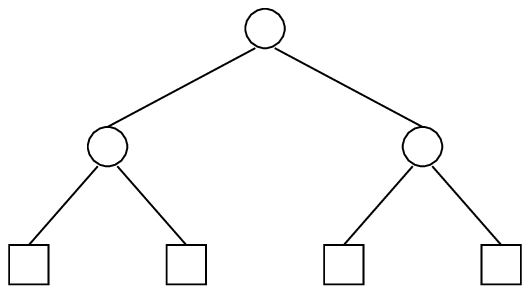} \hspace{2cm} \input{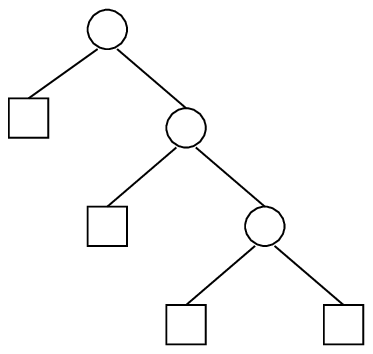}
  \caption{Balanced (left) and Huffman-shaped (right) wavelet trees.}
  \label{fig:wt}
\end{figure}

A rank query $\rank_X(c,r)$ over a wavelet tree is evaluated by a
traversal from the root to the leaf labelled by $c$, as shown in
Fig.~\ref{fig:wt-rank}. The procedure involves rank queries over the
bitvectors stored on the root-to-leaf path.

\begin{figure}
  \centering
  \begin{tabbing}
    10:\=\qquad\=\qquad\=\qquad\=\kill
    \rule{\linewidth}{1pt}\\
    \balgorithm{WT-Rank($c,r$)}\\
    1:\>$v \la \mbox{root}$; $q \la r$ \\
    2:\>\bwhile $v$ is not a leaf \bdo\\
    3:\>\>\bif $c$ is in the left subtree of $v$ \bthen\\
    4:\>\>\>$q \la q - rank_{B(v)}(1,q)$ \\
    5:\>\>\>$v \la \mbox{leftchild}(v)$\\
    6:\>\>\belse\\
    7:\>\>\>$q \la rank_{B(v)}(1,q)$ \\
    8:\>\>\>$v \la \mbox{rightchild}(v)$\\
    9:\>\breturn q\\
    \rule{\linewidth}{1pt}
  \end{tabbing}
  \caption{Rank operation using a wavelet tree.}
  \label{fig:wt-rank}
\end{figure}

There are many data structures for representing bitvectors so that
rank queries can be answered efficiently~\cite{os2007,v2008,cn2008}.
They can be divided into two main categories. Uncompressing techniques
leave the bitvector intact but use a small (usually sublinear) data
structure on top of it. Compressing techniques compress the bitvector
as well as prepare it for rank queries.

\section{Compression Boosting}

Recall that $T$ is a string of length $n$ over an alphabet $\Sigma$ of
size $\sigma$.  For each $c\in\Sigma$, let $n_c$ denote the number of
occurrences of $c$ in $T$.  The zero-order empirical
entropy~\cite{m2001} of $T$ is
\begin{equation}
  \label{eq:entropy}  
H_0(T)=\sum_{c\in\Sigma} \frac{n_c}{n} \log \frac{n}{n_c}
 = \log n - \frac{1}{n} \sum_{c\in\Sigma} n_c \log n_c.
\end{equation}
Let $n_w$ be the number of occurrences of a string $w$ in $T$, and let
$T|w$ be the subsequence of $T$ consisting of those characters that
appear in the \emph{(right) context}~$w$, i.e., that are immediately
followed by $w$. Here $T$ is taken to be a cyclic string, so that each
character has a context of every length.  The $k^\textrm{th}$ order
empirical entropy is
\[
H_k(T) = \sum_{w\in\Sigma^k} \frac{n_w}{n} H_0(T|w)\ .
\]
The value $nH_k(T)$ represents a lower bound on the number of bits
needed to encode $T$ by any compressor that considers a context of
size at most $k$ when encoding a symbol. Note that $H_{k+1}(T)\le
H_k(T)$ for all $k$.

A remarkable property of $L$, the BWT of $T$, is that $T|w$ is a
contiguous substring of $L$ for any $w$; we call the substring the
$w$-context block of $L$. For example, if $T=\str{BANANA\$}$, then
$T|\str{A}=\str{NNB}=L[1..3]$ (see Fig.~\ref{fig:BWT}).  Thus we get
the following result.
\begin{lemm}[\cite{m2001}]
  For any $k\ge 0$, there exists a partitioning of $L_1L_2\cdots
  L_\ell=L$ of the BWT $L$ of $T$ into $\ell\leq \sigma^k$ blocks so
  that
  \[
  \sum_{i=1}^{\ell} |L_i|H_0(L_i) = nH_k(T)\ .
  \]
\end{lemm}
In other words, by compressing each BWT block to zero-order entropy
level, we obtain $k^\textrm{th}$ order entropy compression for the
whole text.  This is called \emph{compression
  boosting}~\cite{fgms2005}.

The space requirement of an FM-index is usually dominated by the
wavelet tree bitvectors.  The total length of the bitvectors in the
balanced wavelet tree of $L$ is $n\lceil\log\sigma\rceil$. Using a
Huffman-shaped wavelet tree reduces this down to at most
$n(H_0(T)+1)$. An alternative way to achieve zero-order compression is
to use compressed bitvector rank indexes. For example, using a rank
index of Raman, Raman and Rao (RRR)~\cite{rrr2007}, the total size of
the rank indexes (without HWT or boosting) is
$nH_0(T)+o(n)\log\sigma$.

Compression boosting improves the $H_0(T)$ factor to
$H_k(T)$~\cite{fmmn2007}: Divide the BWT into context blocks using
context of length $k$ and implement a separate wavelet tree for each
block. There is an additional space overhead of $\Oh(\sigma\log n)$
bits per block from having many blocks and wavelet trees instead of
just one. The total overhead is $o(n)$ for $k\leq
((1-\epsilon)\log_\sigma n)-1$ and any constant $\epsilon>0$.

It may not be optimal to use the same context length in all parts of
$L$. Ferragina et al.~\cite{fgms2005} show how to find an optimal
partitioning with varying context length in linear time. The resulting
compression is at least as good as with \emph{any} fixed $k$.

M\"akinen and Navarro~\cite{mn2007-implicit} show that the boosting
effect is achieved with the RRR bitvector rank index without any
explicit context partitioning. This is called \emph{implicit
  compression boosting}. First, they observe that instead of
partitioning the BWT, we could partion the bitvectors and obtain the
same boosting effect. Second, the RRR technique partitions the
bitvectors into blocks of size $b=(\log n)/2$ and compresses each
independently. The RRR partitioning is not optimal, but M\"akinen and
Navarro show that the overhead due to the inoptimality is at most
$2\sigma\ell b \leq\sigma^{k+1}\log n = o(n)$ under the assumptions
mentioned above.

\begin{thrm}[\cite{fmmn2007,mn2007-implicit}]
  The FM-index either with explicit boosting and optimal
  partitioning~\cite{fmmn2007} or with implicit
  boosting~\cite{mn2007-implicit} can be implemented in
  $nH_k(T)+o(n)\log\sigma$ bits of space for any $k\leq
  ((1-\epsilon)\log_\sigma n)-1$ and any constant $\epsilon>0$.
\end{thrm}

\section{Fixed Block Compression Boosting}

In this section, we show that the compression boosting effect can also
be achieved with a partitioning into blocks of fixed sizes without any
regard to symbol context.

Let $H(x,y)=|B|H_0(B)$, where $B$ is a bitvector containing $x$ 0's
and $y$ 1's. Let $|X|_c$ denote the number of occurrences of a symbol
$c$ in a string $X$. The following lemma shows what can happen to the
total zero-order entropy when two strings are concatenated.  
\begin{lemm}
\label{lm:concat}
  For any two strings $X$ and $Y$ over an alphabet $\Sigma$, 
  \begin{equation*}
    \begin{split}
      0 &\leq
      |XY|H_0(XY) - |X|H_0(X)-|Y|H_0(Y) \\
      &= H(|X|,|Y|)-\sum_{c\in\Sigma}H(|X|_c\,,|Y|_c) \leq H(|X|,|Y|)
      \leq |XY|\ .
    \end{split}
  \end{equation*}
\end{lemm}

\begin{proof}
  The last two inequalities are trivial and the first is a
  standard application of Gibb's inequality. We will prove the
  equality part. For brevity, we write $x=|X|$, $y=|Y|$, $x_c=|X|_c$
  and $y_c=|Y|_c$. Using~(\ref{eq:entropy}), we can write the left-hand
  side terms as follows
  \begin{align*}
  (x+y)H_0(XY) &= (x+y) \log (x+y) 
  - \sum_{c\in\Sigma} (x_c+y_c) \log (x_c+y_c)\\
  xH_0(X) &= x\log x - \sum_{c\in\Sigma} x_c \log x_c \\
  yH_0(Y) &= y\log y - \sum_{c\in\Sigma} y_c \log y_c
  \end{align*}
  and the right-hand side terms as follows
  \begin{align*}
    H(x,y) &= (x+y) \log (x+y) - x \log x - y \log y \\
    H(x_c,y_c) &= (x_c+y_c) \log (x_c+y_c) - x_c \log x_c - y_c \log y_c
  \end{align*}
  From this it is easy to see that the terms on both sides match.
\qed
\end{proof}
In other words, the concatenation cannot reduce the total entropy, and
the entropy can increase by at most one bit per character.
Furthermore, the maximum increase happens only if the two strings have
the same length and no common symbols.

Using the above lemma we can bound the increase in entropy when we
switch from a context block partitioning to a fixed block
partitioning.
\begin{lemm}
  \label{lm:partition}
  Let $X_1X_2\cdots X_\ell=X$ be a string partitioned arbitrarily into
  $\ell$ blocks. Let $X^b_1X^b_2\cdots X^b_m=X$ be a partition of $X$
  into blocks of size at most b. Then
  \[
  \sum_{i=1}^m |X^b_i| H_0(X^b_i) \leq \sum_{i=1}^\ell |X_i|
  H_0(X_i) + (\ell-1) b\ .
  \]
\end{lemm}

\begin{proof}
Consider a process, where we start with the first partitioning, add
the split points of the second partitioning, and then remove the split
points of the first partitioning (that are not split points in the
second). By Lemma~\ref{lm:concat}, adding split points cannot increase
the total entropy, and removing each split point can increase the
entropy by at most $b$ bits. \qed
\end{proof}
If we assume the same number of blocks in the two partitionings, the
very worst case increase in the entropy is $n-b$ bits. 
However, such a worst case is very unlikely and in practice
the increase is much smaller.

If we set the block size to $b=\sigma(\log n)^2$, we obtain the
following result.

\begin{thrm}
  The FM-index with explicit boosting and blocks of fixed sizes
  can be implemented in
  $H_k(T)+o(n)\log\sigma$ bits of space for any $k\leq
  ((1-\epsilon)\log_\sigma n)-1$ and any constant $\epsilon>0$.  
\end{thrm}

\begin{proof}
  Using context block boosting with fixed context length $k$ and RRR
  to compress the bitvectors, the size of the FM-index is
  $nH_k(T)+o(n)\log\sigma$ bits. When we switch from context blocks to
  fixed blocks, we must add two types of overhead. First, by
  Lemma~\ref{lm:partition}, the total entropy increases by at most
  $\sigma^kb=\sigma^{k+1}(\log n)^2=n^{1-\epsilon}(\log n)^2=o(n)$
  bits. Second, the space needed for everything else but the bitvector
  rank indexes is $\Oh(\sigma\log n)$ bits per block. In total, this
  is $\Oh(n/\log n)=o(n)$ bits. Thus the total increase in the size of
  the FM index is $o(n)$ bits.
\end{proof}
Thus, we have the same theoretical result as with context block
boosting or implicit boosting.  

The advantages of fixed block boosting
compared to context block boosting are:
\begin{itemize}
\item To compute $\rank_L(c,r)$, we have to find the block containing the
  position $r$. With fixed blocks this is simpler and faster than with
  varying size context blocks.
\item Computing the optimal partitioning is complicated and
  expensive in practice. With fixed blocks, construction is much
  simpler and faster.
\end{itemize}
Explicit boosting (with either context blocks or fixed blocks) enables
faster queries than implicit boosting for the following reasons:
\begin{itemize}
\item Compressed bitvector rank indexes are slower than uncompressed
  ones by a significant constant factor. Explicit boosting can achieve
  higher order compression with Huffman-shaped wavelet trees allowing
  the use of the faster uncompressed rank indexes.
\item With implicit boosting, i.e., with a single wavelet tree for the
  whole BWT, the average count query time for a pattern $P$ is
  $\Theta(|P|\log\sigma)$ with a balanced wavelet tree and
  $\Theta(|P|H_0(T))$ with a HWT.
  With explicit boosting and HWTs, the average query time is reduced down to
  $\Oh(|P|H_k(T))$.
\end{itemize}

\section{Experimental Results}
\label{sec-experiments}

To assess practical performance we used the files listed in
Table~\ref{tab-data}\footnote{Available from
  \url{http://pizzachili.dcc.uchile.cl/}.}.
All tests were conducted on a 3.0 GHz Intel Xeon CPU
with 4Gb main memory and 1024K L2 Cache.
The machine had no other significant CPU tasks running.
The operating system was Fedora Linux running kernel 2.6.9. The
compiler was g++ (gcc version 4.1.1) executed with the {-}O3 option.
The times given were recorded with the C {\tt getrusage} function.
The memory requirements are sums of the sizes of all data structures
as reported by the {\tt sizeof} function.

\begin{table}[tb]
\caption{
\label{tab-data}Data sets used for empirical tests. For each type
of data (\dna, \xml, \english, \sources) a 100Mb
file was used.
}
\begin{center}
\tabcolsep 10pt
\begin{tabular}{l r r r}
\hline
Data set name &
\multicolumn{1}{c}{$\sigma$} &
\multicolumn{1}{c}{$H_0$} &
\multicolumn{1}{c}{mean LCP} \\
\hline
\xml        & 97&5.23&   44\\
\dna        & 16&1.98&   31\\
\english    &239&4.53&2,221\\
\sources    &230&5.54&  168\\
\hline
\end{tabular}
\end{center}
\end{table}

\begin{figure}[tb!]
\centering
\subfloat{\includegraphics[scale=0.8]{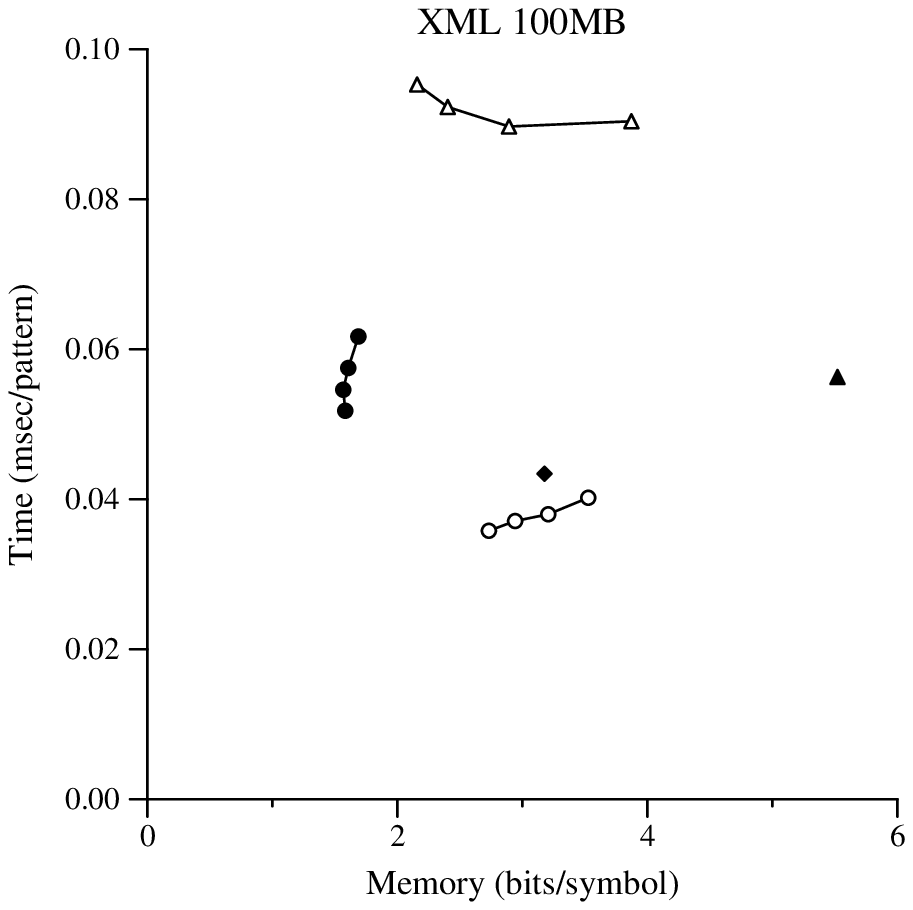}}
\subfloat{\includegraphics[scale=0.8]{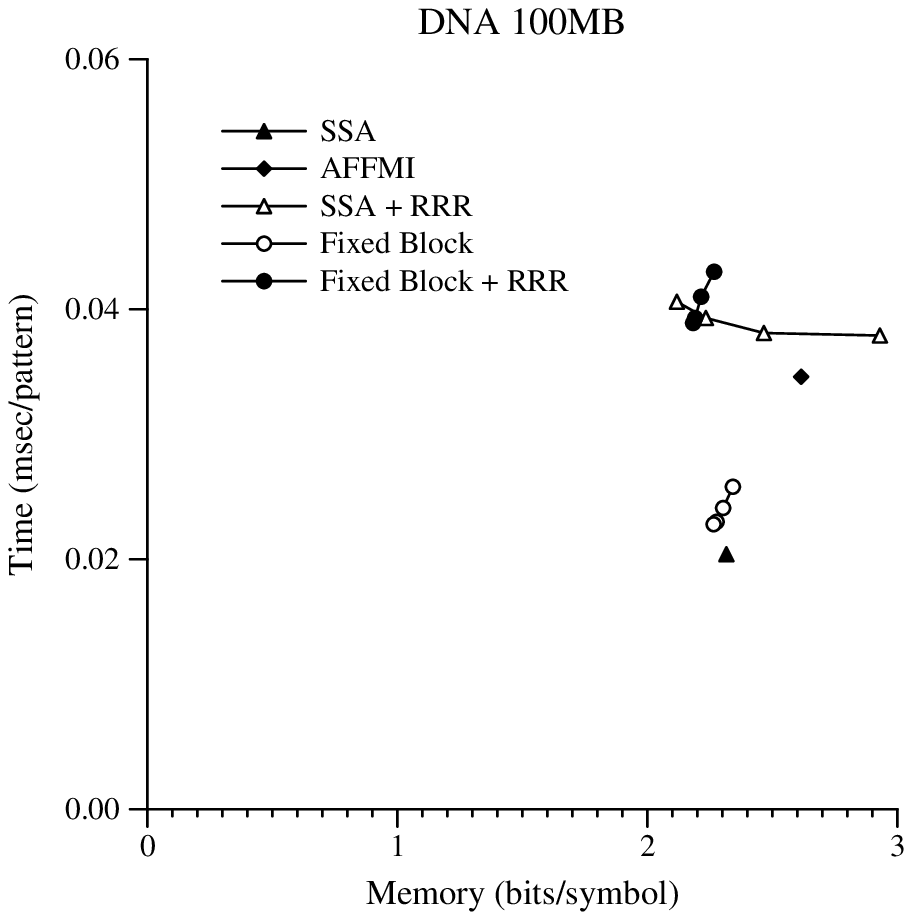}}

\subfloat{\includegraphics[scale=0.8]{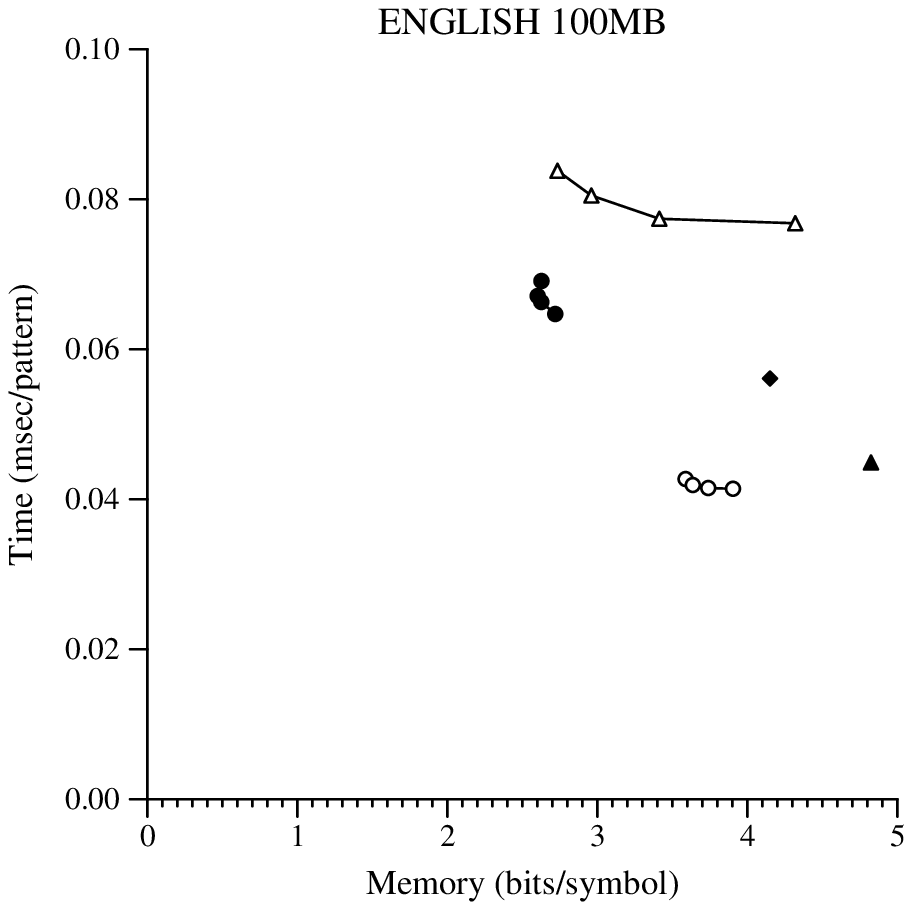}}
\subfloat{\includegraphics[scale=0.8]{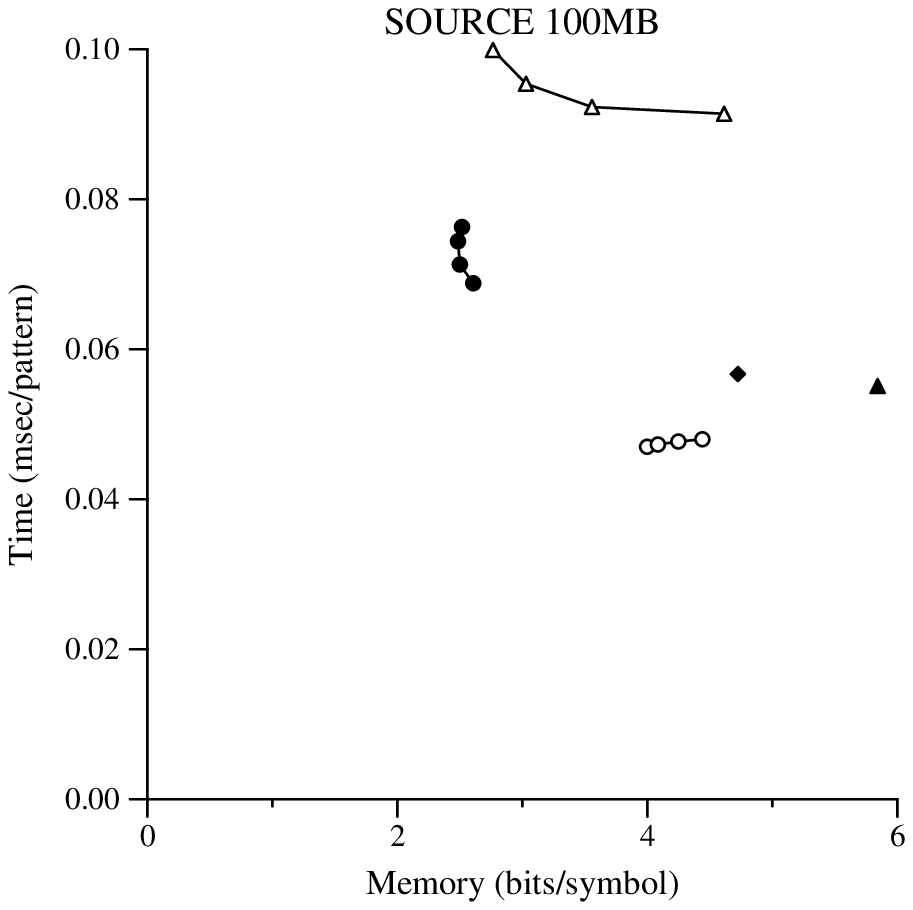}}
\caption{\label{fig-tradeoff} Time-Space tradeoff for various
  self-indexes. Memory (abcissa) is the index size in bits per input
  symbol. Time (ordinate) is the average number of milliseconds taken
  to count pattern occurrences (averaged over $10^6$ patterns).  }
\end{figure}

We measured the following FM-Index variants:
\begin{itemize}
\item SSA~\cite{mn2005} simply stores a single HWT for the whole $L$,
  consuming $nH_0 + o(n\log\sigma)$ bits of space. This is the fastest
  index according to experiments in both~\cite{cn2008}
  and~\cite{fgnv2009}.
\item SSA+RRR is the implicit compression boosting approach of
  M{\"a}kinen and Navarro~\cite{mn2007-implicit}.  As with SSA it
  builds a single HWT of $L$, however the bitvectors of the wavelet
  tree are now stored in a RRR compressed rank data structure. This
  method was first implemented by Claude and Navarro~\cite{cn2008}.
\item AFFMI~\cite{fmmn2007} uses optimal context block boosting with a
  separate HWT for each block. The
  implementation we use is from~\cite{fgnv2009}.
\item Fixed Block and Fixed Block+RRR are implementations of the new
  fixed block boosting technique that use, respectively, plain and RRR
  preprocessed HWTs to represent blocks.
\end{itemize}

Figure~\ref{fig-tradeoff} shows the trade-off between index size and
pattern counting time.  Following the methodology
of~\cite{cn2008,fgnv2009} we report query times averaged over a large
number of random patterns, extracted from the underlying text.

With the compressible texts (\xml, \sources and \english) the fixed
block indexes dominate the others in both space and time. On \dna,
which is not very compressible, fixed block indexes are still small
and fast, but the ranks stored at block boundaries are no longer paid
for by compression and the SSA+RRR, which does not need to store ranks
at block boundaries, is the smallest index. The small alphabet of \dna
means the single HWT of the SSA is shallow, making it fast.

The AFFMI, despite using optimal partitioning, is larger and slower
than the fixed block indexes. AFFMI stores a bitvector marking the
partitioning and issues a rank query on this bitvector to determine
the appropriate wavelet tree to use at each step in the backward
search process. This adds a significant time and space overhead which
the fixed block approach avoids entirely.

\section{Concluding Remarks}

The indexes we have presented based on fixed block compression
boosting are the most practical self-indexes to date, but we believe
there is yet more room for improvement. Our current focus is on
improving the RRR data structure to better exploit the structure of
wavelet tree bitvectors produced by the BWT. We are also exploring an
improved implementation of Huffman-shaped wavelet trees which use
substantially less space, enabling smaller blocks and thus better
compression.

A virtue of fixed block compression boosting our experiments have not
touched on is construction, which is easier with fixed blocks. The
final phase of indexing, where the BWT is turned into an FM-index, now
requires only $nH_k + o(n)\log\sigma + b\log\sigma$ bits of space,
instead of the (at least) $n\log\sigma + o(n)\log\sigma$ bits required
by variants to date. If the final index does not have to reside in
memory then at most $2b\log\sigma$ bits are needed for construction of
the index from the BWT. Construction time remains linear and is fast
in practice as the BWT is scanned only once, from left to right. For example, 
construction a fixed block index for the XML file takes just 12 seconds, 
while to build an index with optimal compression boosting requires 273 seconds. 
Ease of construction is also important when the aim is full inversion
of the BWT in a general purpose file compressor~\cite{kp2010}.
\bibliographystyle{splncs_srt}
\bibliography{sa}

\end{document}